\documentclass[12pt]{article}

\usepackage{graphicx}
\usepackage{subcaption}
\usepackage{enumerate}
\usepackage[T1]{fontenc}
\usepackage[utf8]{inputenc}
\usepackage[english]{babel}
\usepackage[pdfa=true]{hyperref}
\usepackage{amsmath}
\usepackage{amssymb, amsthm, amsfonts}
\usepackage{dsfont}
\usepackage{layout}
\usepackage{fancyhdr}
\usepackage{relsize}
\usepackage{setspace}
\usepackage{endnotes}
\let\footnote=\endnote
\usepackage{authblk}
\usepackage[
  backend=biber,
  style=authoryear,
  natbib=true,
  maxnames=3,
  maxcitenames=3,
  giveninits=true,
  uniquename=init,
  dashed=false
]{biblatex}
\usepackage{csquotes} 

\theoremstyle{plain} 
\newtheorem{theorem}{Theorem}
\theoremstyle{definition}
\newtheorem{definition}[theorem]{Definition}

\title{A Comparative Analysis of Modeling Approaches for the Association of FAIR Digital Objects Operations}
\author[1]{Nicolas Blumenröhr\thanks{\texttt{nicolas.blumenroehr@kit.edu}}}
\author[2]{Jana Böhm}
\author[1]{Philipp Ost}
\author[3]{Marco Kulüke}
\author[4]{Peter Wittenburg}
\author[5]{Christophe Blanchi}
\author[2]{Sven Bingert}
\author[2]{Ulrich Schwardmann}

\affil[1]{Karlsruhe Institute of Technology, Scientific Computing Center, Germany}
\affil[2]{GWDG, Germany}
\affil[3]{German Climate Computing Center, Germany}
\affil[4]{Max Planck Institute for Psycholinguistics, Netherlands}
\affil[5]{DONA Foundation, Switzerland}
\date{}
\begin{document}

\maketitle
\begin{abstract}
    The concept of FAIR Digital Objects represents a foundational step towards realizing machine-actionable, interoperable data infrastructures across scientific and industrial domains. As digital spaces become increasingly heterogeneous, scalable mechanisms for data processing and interpretability are essential. This paper provides a comparative analysis of various typing mechanisms to associate FAIR Digital Objects with their operations, addressing the pressing need for a structured approach to manage data interactions within the FAIR Digital Objects ecosystem. By examining three core models -- record typing, profile typing, and attribute typing -- this work evaluates each model's complexity, flexibility, versatility, and interoperability, shedding light on their strengths and limitations. With this assessment, we aim to offer insights for adopting FDO frameworks that enhance data automation and promote the seamless exchange of digital resources across domains.
\end{abstract}
\section{Introduction}
\label{sec:Introduction}
Science advances to incorporate previously isolated domains into more comprehensive views of physical and biological processes, and industry addresses an increased need for interconnected supply chains and operational technologies. Therefore, it becomes increasingly important to develop common approaches for automated acquisition, interpretability, and processing of digital data~\citep{wilkinson_fair_2016, directorate-general_for_research_and_innovation_european_commission_eosc_2021, jeffery_not_2021}. Processing very large, heterogeneous, and diverse data sets from different domains using the existing sets of incompatible APIs applicable to those data sets is simply not possible~\citep{soiland-reyes_evaluating_2024}. However, it is widely agreed that the future of data processing must be highly automated to cope with the increasing amounts of digital resources which are of great importance for meeting the requirements of UNs Sustainable Development Goals~\citep{madavarapu_ai-powered_2024}. A foundational infrastructure that provides a common and more automatable approach to discovering and executing operations on data could have the same impact on data processing that the Internet and Web technologies have had on communication and multimedia information exchange~\citep{schultes_fair_2019, peter_wittenburg_common_2018}. This could lead to large and necessary advances in scientific discovery, and industrial efficiency and sustainability.

The FAIR Digital Objects (FDOs) concept describes how such an infrastructure could be realized by representing digital resources of any type in a way that enables automated processing~\citep{schultes_fair_2019, smedt_fair_2020, blumenrohr_fair_2024}. It does so by implementing the FAIR Principles~\citep{wilkinson_fair_2016}, which provide guidelines for better data management and stewardship, using the Digital Object framework~\citep{kahn_framework_2006}. While different implementation strategies for FDOs exist, they all aim towards an automated processing by the machine-actionable characteristics of an FDO that is enabled by operations. An operation will in general be associated with an FDO by its typing mechanism and may be executed on different FDO levels, i.e., the metadata or the bit sequence of the digital resource~\citep{blumenrohr_fair_2024}. Operations may range from basic CRUD operations to more advanced operations and can be implemented using various technologies. However, the exact specification of a type system for FDOs that enables a mechanism to associate the objects with applicable operations is not yet fully scoped~\citep{soiland-reyes_evaluating_2024, blumenrohr_fair_2024}. At this point, there exist different views and implementations for associating FDOs and operations by typing. In fact, having multiple approaches is desirable, since there may not be a one-size-fits-all solution. Nevertheless, to ensure an interoperable ecosystem for FDOs, it is important to assess if and how these approaches are compatible with each other, also with respect to client knowledge. Providing a structured analysis of these association models will support the adoption of FDOs by different communities. Associating FDOs with their operations is seen as the missing step in data processing automation by machine-actionability. Formalized type specifications and user intentions paired with formalized reuse conditions will be key in this regard.
 
In this work, we provide an assessment of existing typing mechanisms for associating  FDOs with their operations based on different conceptual data models. We describe each data model along with an implementation example, and comparatively evaluate their characteristics with respect to these typing mechanisms. Based on the evaluation, we discuss the results in the larger context of FDO processability and perspectives for communities that want to adopt the concept.

\section{Background}
\label{sec:Background}
\subsection{Foundations of FDOs and the Core Model}
\label{sec:Foundations of FDOs and the Core Model}
FDOs are high-level, atomic, and persistent entities that bundle information for FAIR processing of a bit-sequence including different kinds of metadata, are referenced by a Persistent Identifier (PID), are FAIR compliant in their core mechanisms~\citep{, blumenrohr_fair_2024}, and can be protected against misuse in various dimensions~\citep{smedt_fair_2020}. In the FDO core model, each FDO represents a basic structure which allows for different configurations, i.e. configuration types~\citep{larry_fdo_2022}, and has the following characteristics: 
\begin{itemize}
    \item A Handle PID will be resolved into an FDO information record that contains the Kernel Information.
    \item The Kernel Information describes the FDO core metadata attributes such as its data type, location and additional metadata references. 
    \item The Kernel Information is structured as a set of attributes expressed as a set of key-value pairs, aggregated by a Kernel Information Profile~\citep{weigel_2019_3581275} that the information record must conform to.
    \item For compatibility reasons, only a minimal set of attributes are being specified in the Kernel Information Profile as also proposed by the FDO Forum\footnote{\url{https://fairdo.org/}} and the Research Data Alliance \footnote{\url{https://www.rd-alliance.org/}}. 
    \item Each attribute being included in the profile must be defined and registered in a public registry according to the specification of PID-Information Types~\citep{schwardmann_automated_2016}, making it machine-interpretable.
    \item It is actionable through a set of operations that are associated with the Kernel Information via a typing mechanism.
\end{itemize}

This minimal definition of the FDOs follows the spirit of the Internet which defines a basic package structure for information transfer and allows making use of a communication protocol for FDOs, the Digital Object Interface Protocol (DOIP)~\citep{dona_foundation_digital_2018}. FDOs can represent bit-sequences with different kinds of content such as data, metadata, configurations, semantic assertions, software, etc. As illustrated in \autoref{fig:fdo_core_model}, due to their generic core model, FDOs have the potential to be used as a basic interoperability layer to connect different types of repositories and data spaces~\citep{curry_data_2022}. For further technical details on FDOs, we refer to the FDO Overview~\citep{ivonne_2023_7824714}, and the FDO Requirement Specifications~\citep{ivonne_2023_7782262}. Note that the term \textit{profile} is used interchangeably with the term \textit{Kernel Information Profile} in the subsequent sections.

\begin{figure}[tbp]
    \centering
    \includegraphics[width=0.6\textwidth]{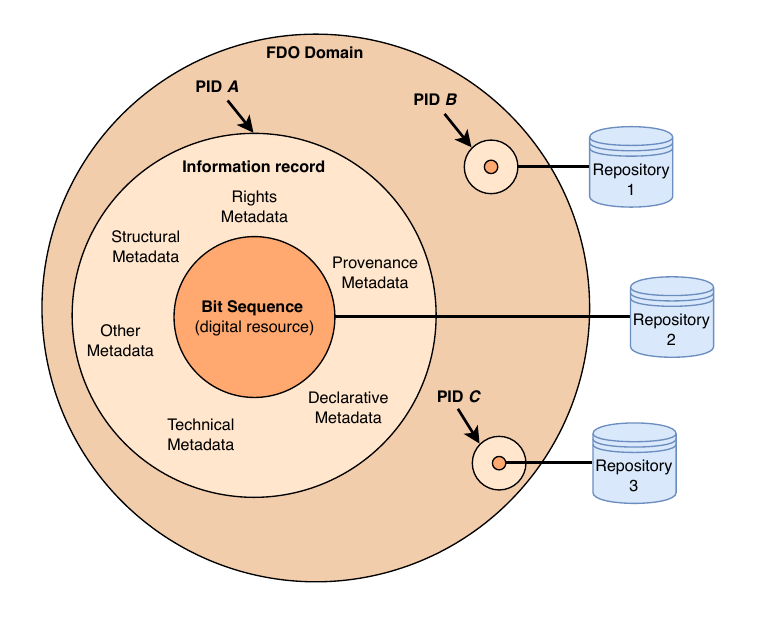}
    \caption{The FDO Core model.}
    \label{fig:fdo_core_model}
\end{figure}

\subsection{Problem Description}
\label{sec:Problem Description}
Several works on FDO implementations have described the theoretical applicability of FDO operations, e.g.~\citep{blumenrohr_fair_2024, blanchi_canonical_2022, islam_fair_2023, lannom_fair_2020} or have even implemented specialized systems that enable the execution of operations in their FDO ecosystem, e.g.~\citep{islam_assessing_2023, blumenrohr_implementation_2023}. However, to the best of our knowledge, a set of generic mechanisms for associating these operations with FDOs via a set of rules, i.e., a type system, in compliance with the description of the original concept has not been worked out yet. This makes it hard to assess and reproduce these use-case specific operation frameworks.
 
The authors of this work have developed typing mechanisms to associate FDOs and operations within their organizations, which were extensively discussed in the frame of the FDO Forum. At this point, there exist some reference implementations for these mechanisms as described in the following sections, but no detailed explanations on their data models and how they compare to each other. Therefore, we see this paper as a step forward to assess these association models and provide a baseline for implementing (inter-)operable FDO ecosystems.

\section{Models for Associating FDOs with their Operations}
\label{sec:Models for Associating FDOs with their Operations}
In this section, we will first describe the different modeling approaches for the association of FDOs and operations and their underlying typing mechanisms. We thereby assume that an FDO is specified according to the core model described in section \ref{sec:Foundations of FDOs and the Core Model}. We first elaborate on the general idea of these typing mechanisms that we define as part of a type system for FDOs that defines rules for the FDO components. These typing mechanisms are related to well-known typing principles in computer science and are finally incorporated in each association model. Technical implementation details for these association models are not considered.
 
In the second part of this section, we will go through several application examples that use the different association models based on the typing mechanisms introduced under the rules of the FDO type system.

\subsection{Typing Mechanisms}
\label{sec:Typing Mechanisms}
The problem with the terms type and typing is that they are generic, and often have different definitions across disciplines and technologies. This work does not aim towards providing a definition of these terms but requires a more concrete description in the context of FDOs. It can be said at this point that many of the terms employed relate to ideas known from Object-oriented Programming (OOP), of which relations to other principles such as abstraction and encapsulation have already been described by the works of~\citep{schultes_fair_2019, blumenrohr_fair_2024}. The next step is to also infer mechanisms for associating operations on the basis of abstraction and encapsulation provided by FDOs. It is important to note that we consider the analogy between OOP and FDOs only on an abstract, conceptual level, whilst the implementation details of FDOs are a different aspect. The following terms also found in OOP are therefore defined in the context of FDOs:
\begin{itemize}
    \item Abstraction and Encapsulation: FDOs pack data and metadata into a single unit by definition, encapsulating internal details. The interface to the FDO is given by attributes describing possible interactions. The set of attributes is given by its profile. The profile itself is therefore a class. It is an abstraction of all FDOs that fulfill the requirements of the profile.
    \item Type system: inspired by the work of~\citep{citeulike:105547}, we define this as a set of rules for validating how FDOs are typed and associated with a set of operations by one or more typing mechanisms.
    \item Typing Mechanism: the exact procedure of how the FDO kernel information, consisting of the profile, information record and its content that make up the FDO type, is used to determine if an operation is associated with a particular FDO.
\end{itemize}
The typing mechanisms for FDOs to associate operations with FDOs and their characteristics are described in the following. The details and relations of these procedures to principles known from OOP are illustrated in \autoref{fig:typing_mechs}.
\begin{figure}[tbp]
    \centering
    \includegraphics[width=0.9\textwidth]{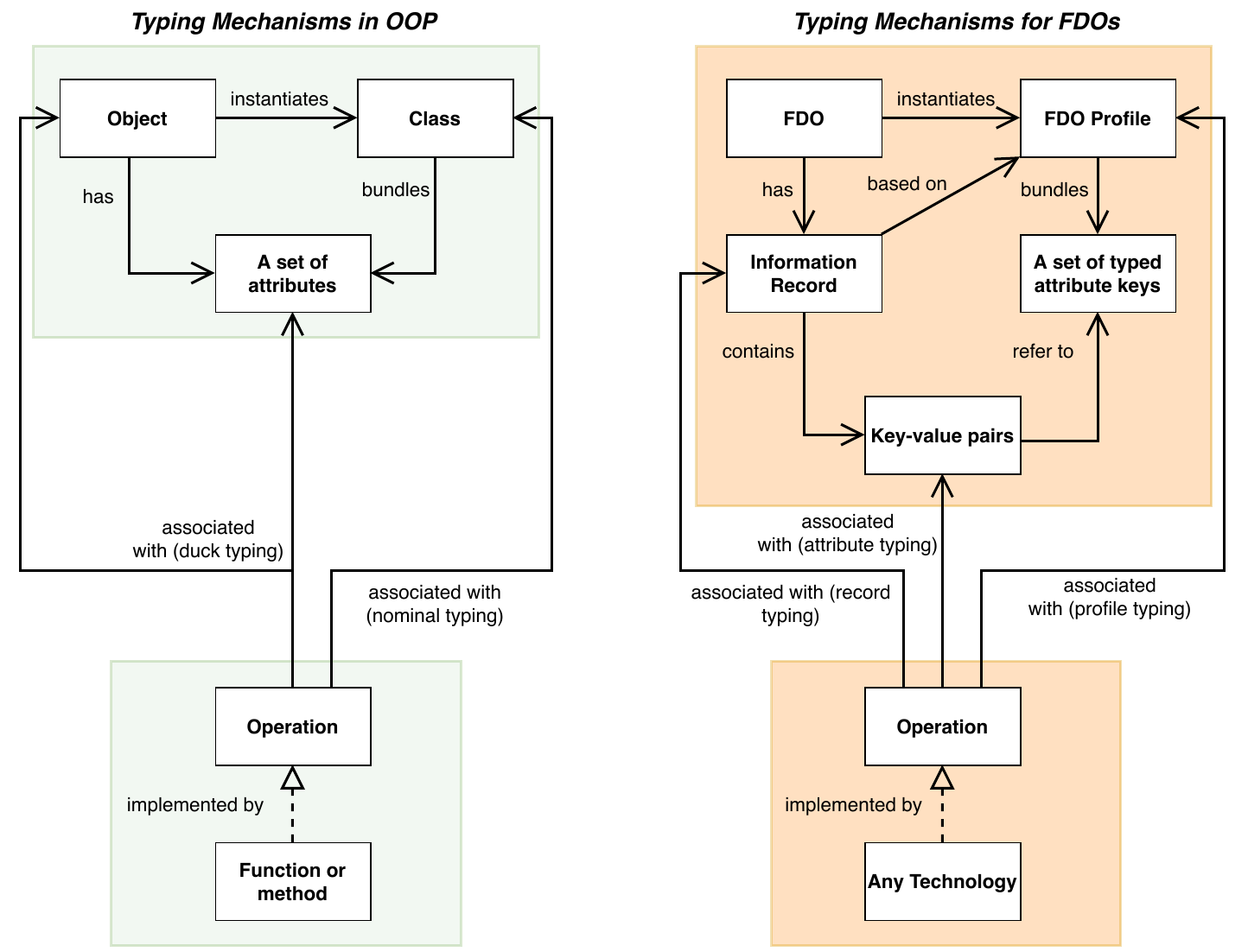}
    \caption{Typing Mechanisms. The conceptual typing mechanism to associate FDOs and their operations in analogy to OOP.}
    \label{fig:typing_mechs}
\end{figure}

With respect to the association approach, there exist two obvious possibilities. The first is to extend the FDO interfaces and include operations as attributes into the FDO record by changing the profile (operation association to FDO). The other is to leave the interfaces of the data FDOs unchanged, and to describe requirements of operations for the interfaces of the operation representation (FDO association to operation). This can be represented as relations between operations and types in a dedicated type system, e.g. based on PID-Information Types. 

Even though object association to operations is also possible within OOP, operation association to objects of classes is encouraged there as part of the encapsulation. Operations behind REST services are also usually associated to the objects behind their interfaces. Object association to operation is more commonly used in the context of media types, where the applicability of an operation is decided by the type of object. The type encapsulates the internal complexity of both the object and the operation.
This results in three core mechanisms of typing that we detail in the following.
\subsubsection{Record Typing}
\label{sec:Record Typing}
The most straightforward way of typing in FDOs can be achieved by specifying an operation directly in the information record of the FDO, thereby directly associating each operation with the individual object. The type is hereby purely defined by the constellation of applicable operations. Conceptually, this is similar to the principle of duck typing in OOP where the type of an object is determined not by its explicit class but by the methods it supports, focusing on what the object can do rather than what it is. All applicable operations are therefore also part of the attributes in the FDO information record and fixed at instantiation time of the object.

\subsubsection{Profile Typing}
\label{sec:Profile Typing}
Another way is profile typing which means that operations that are associated with an FDO are inferred by their association with the profile that is instantiated by this FDO and is thus considered the type. Attaching the operations to FDO profiles is possible because each FDO has a profile as a mandatory attribute in its information record according to the Kernel Information requirements. This is comparable to nominal typing in OOP where an operation in the form of a method is bound to a class, meaning that it operates on instances of that class (objects) and has access to the class's attributes. 

\subsubsection{Attribute Typing}
\label{sec:Attribute Typing}
This typing mechanism uses the set of attributes in an FDO’s information record, where each operation is associated by the presence of one or more attributes that constitute the type in dependency of these requirements. This also relates to duck typing in OOP with the aspect that an object's usability can also be determined by the presence of specific attributes at runtime, rather than the object's class. 
This works for FDOs because their attributes refer to the specification of PID Information Types, meaning that each element is unambiguously identified and can be reused and recognized for all FDOs.

\subsection{Implementation Examples}
The examples described in this subsection originate from different projects and organizations the authors are involved in, using different types of data, technologies and service architectures. We concentrate here on the association models and the essential workflow, also considering information exchange between FDO services and the client side. Therefore, apart from a minimal necessary description, we do not provide technical details of each implementation and the service components that are used in these projects. We also do not further explain the details of how these operations are ultimately applied to the contents of the FDOs they are associated with. For this, we refer to the references provided in each section. We also want to point out that different complexity levels of these implementations are not necessarily related to the complexity of the individual association model. These will be evaluated in section \ref{sec:Model Evaluation and Discussion}. However, according to the FDO core model, each FDO in these examples is registered at - and thus resolvable via - the Handle Registry, has a typed information record, and complies to one of the known FDO configuration types. 
\subsubsection{Record Typing in Interactive Computing Environments}
\label{sec:Record Typing in Interactive Computing Environments}
This example considers a simple FDO information record that represents a catalog containing links to various climate model simulations described by domain-specific metadata key-value pairs. FDO-related information is statically implemented in the record. Hence, \autoref{fig:record_typing} lays out how the implementation of an association mechanism for operations via record typing works in principle.
\begin{figure}[tbp]
    \centering
    \includegraphics[width=0.7\textwidth]{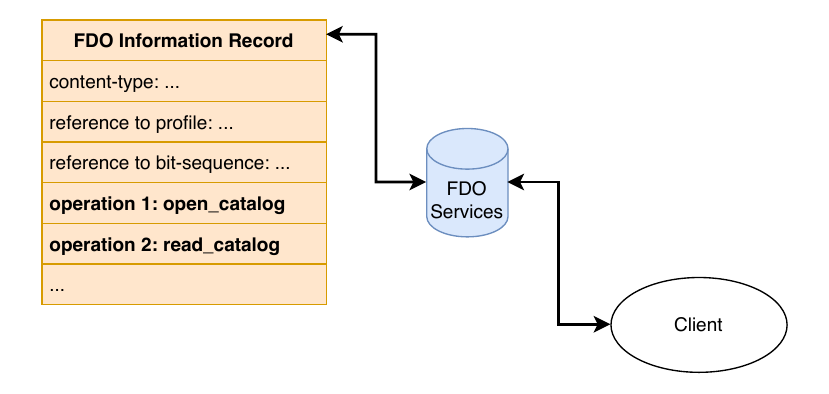}
    \caption{Record typing example. The conceptual workflow for interacting with an FDO based on record typing.}
    \label{fig:record_typing}
\end{figure}
The diagram shows a workflow illustrating the interaction between an FDO and a client using an computational environment, i.e., a Jupyter Notebook, to retrieve predefined operations (here labeled as operation 1 for opening the catalog and operation 2 for reading the catalog) that are bundled in the information record among other metadata required to execute the operation such as the content-type, the  reference to the bit-sequence, or other metadata.

Depending on the FDO service, a client can either request the list of associated operations or directly retrieve them from the FDO’s information record. In this example, these operations are a specification of code that is executable on the client side.

From a technical implementation perspective, the information record itself contains a set of operations that are in principle relevant to any object conforming to the content-type it represents. 
Note that the FDOs cannot dynamically change their operations or substitute them at runtime. The operations are fixed and cannot vary based on different FDO subtypes. The Jupyter Notebook can be found at~\citep{kuluke_2025_14860533}.

\subsubsection{Profile Typing with Multiple Registries}
\label{sec:Profile Typing with Multiple Registries}

Within the FDO One project,\footnote{\url{https://fdo-one.org}} the focus is on providing basic operations for FDOs to build up a functional FDO ecosystem, e.g. CRUD operations (create and delete an FDO, get or update the (meta)data of an FDO) or copying an FDO and moving a distributed FDO from one storage location (data service) to another. For those type of operations, domain-specific attributes and content-types of bit-sequences are irrelevant. Rather, the structure of the FDO itself is of importance, for example, whether it represents zero, one, or multiple (meta)data bit-sequences and how those are stored. This information is determined by the FDO profile. Hence, the profile typing mechanism is used to associate those operations to FDO profiles. In particular, each FDO profile does not only contain a list of mandatory and optional attributes which must be present in an FDO information record, but also a list of operations that can be applied to any FDO complying with this specific FDO profile. Profiles are registered in the profile registry, which is based on a Data Type Registry\footnote{\url{https://typeregistry.lab.pidconsortium.net/}}.

As outlined in \autoref{fig:profile_typing}, to find operations that are associated to an FDO, a client may retrieve the profile (either directly or through a software component) and receive a list of PIDs identifying operations that are associated to this FDO. The operations, in turn, are registered in the operation registry together with all necessary execution information\footnote{Strictly speaking, in order to determine whether the above mentioned operations are executable on an FDO, the technical capabilities of the data service where the FDO is stored also need to be taken into account. This is because these type of operations do not just involve reading bit-sequences, but potentially also require to register new PIDs or manipulate FDO information records, which requires actions of the data service itself. Hence, a service registry is used to store which specific profiles and operations a data service supports. However, describing these mechanisms in detail is out of the scope of this paper.}.
For further reading and technical details of the FDO One testbed implementation, we refer to~\citep{fdo-one-testbed-doc}.
\begin{figure}[tbp]
    \centering
    \includegraphics[width=0.9\textwidth]{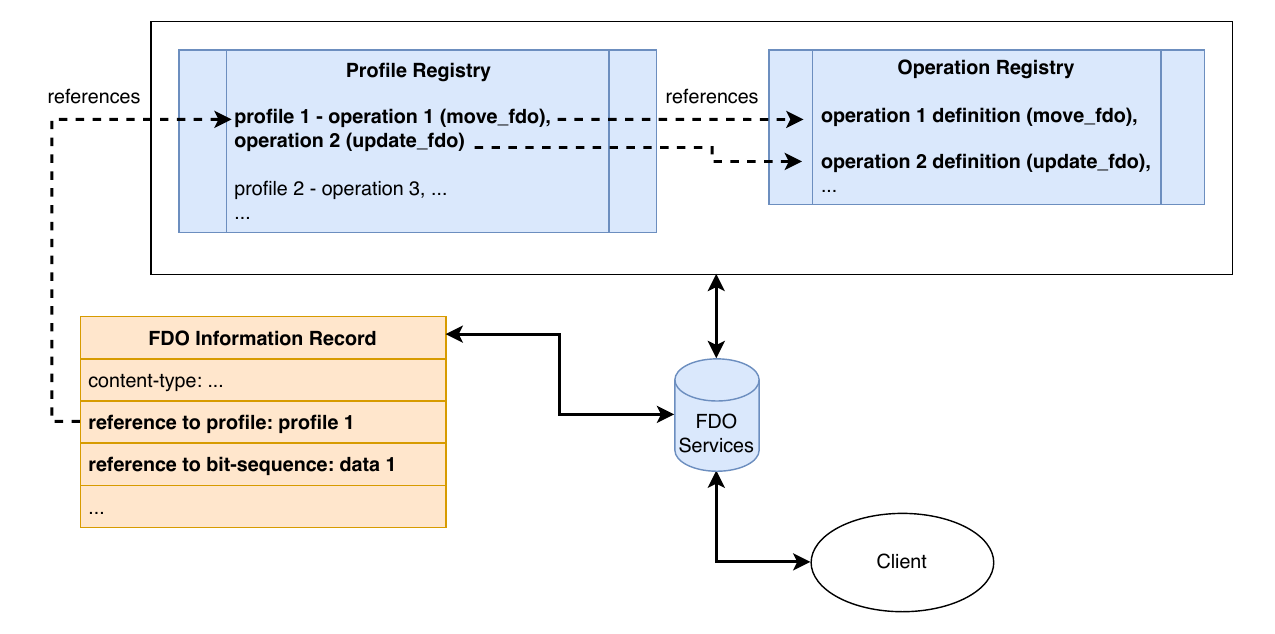}
    \caption{Profile typing example. The conceptual workflow for interacting with an FDO based on profile typing. Irrespective of the service architecture that is used to implement and execute operations such as the three registries in this example, the FDO service must infer the association between the profile of an FDO and its set of operations.}
    \label{fig:profile_typing}
\end{figure}

\subsubsection{Attribute Typing with Operation FDOs}
\label{sec:Attribute Typing with Operation FDOs}
To realize the attribute typing mechanism, an operation must be represented in a way that allows it to be related to the attributes in the targeted FDO's information record that represents research data (i.e., labeled here as \textit{data FDO}). This could be easily provided by representing the operation itself as an FDO as well, which we label here as \textit{operation FDO}. This follows the concept's generic approach that each type of bit-sequence can be represented as FDO. The specific implementation of the operation is thus described in this operation FDO information record, detailing its implementation, possible execution mechanism and the type-association requirements in the form of an attribute's key-value pair.

In this modeling approach, the association can be determined by considering one or more attributes, validating either only their key presence or the presence of key-value pairs. An example for this mechanism is illustrated in \autoref{fig:attribute_typing}, where a data FDO and two operation FDOs are shown. Each operation FDO represents the implementation of the underlying operation that is applied either to the bit sequence, i.e., operation 1 for schema validation, or to the kernel metadata, i.e., operation 2 for license evaluation. The information record of the operation FDO contains at least one key-value pair where the key expresses the \emph{requiredInput} and the value references the attribute type that indicates applicability of the operation to all FDOs that contain this attribute type in their information record. Depending on these requirements, only the key to the referenced attribute type, or the key and a specific value in the form of a tuple (cf. operation 1) may be specified. 
This construction can be seen as a dynamic typing mechanism, where operations are “aware” of the traits an FDO must have in order to be applicable to them. With respect to the infrastructure, there will be additional services required that know how to interpret and validate these type-based relations and subsequently execute the implemented operation, which is not detailed in this work. For further reading and technical details of this example, we refer to~\citep{nicolasblumenroehr_2025_14886300}. 
\begin{figure}[tbp]
    \centering
    \includegraphics[width=0.9\textwidth]{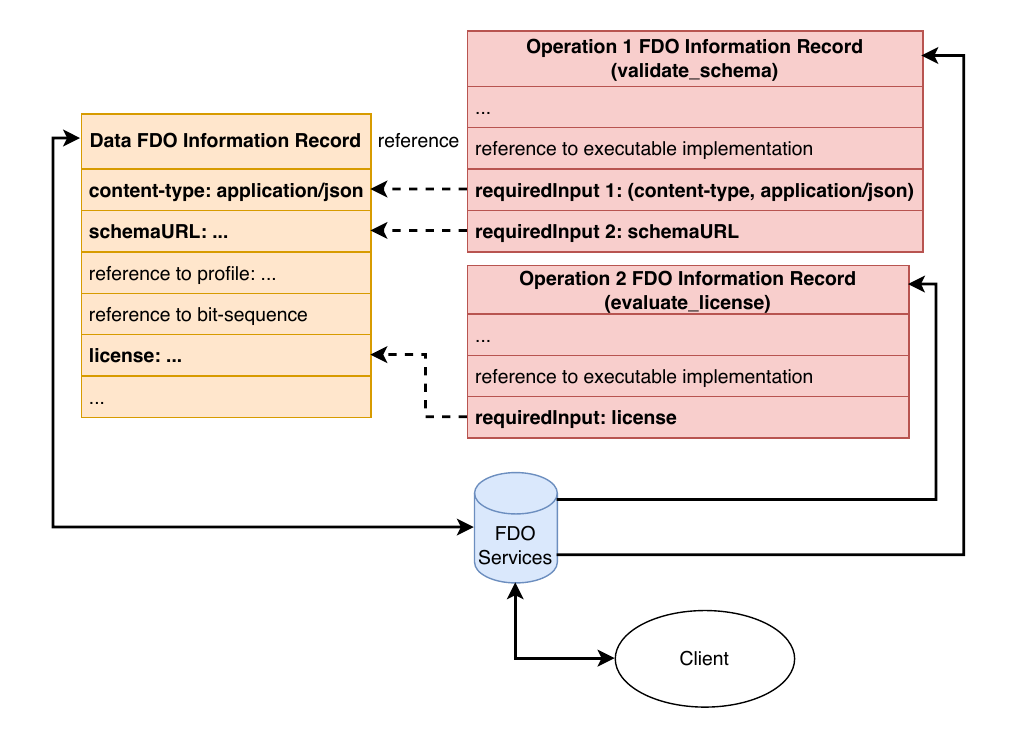}
    \caption{Attribute Typing example. The conceptual workflow for interacting with an FDO based on attribute typing. Irrespective of how the operation is finally performed (requested by the service in this example), the FDO service must infer the association based on the information record contents and reference of the data- and operation FDOs. }
    \label{fig:attribute_typing}
\end{figure}

\section{Model Evaluation and Discussion}
\label{sec:Model Evaluation and Discussion}
In order to evaluate the different approaches for associating FDOs with operations based on the typing mechanisms, we embed the association approaches into a mathematical context by modeling them as directed graphs (section \ref{sec:graph model}). Afterwards, a set of quality indicators is defined that are inspired by the methods used in the domain of entity-relationship modeling as described by~\citep{moody_metrics_1998} (section \ref{sec:quality measures}). These quality indicators finally serve the purpose of putting the different association models in relation to each other and evaluating their advantages, disadvantages and compatibility. To quantify the differences, we define metrics for these quality indicators that are evaluated on each graph model separately. We will also consider qualitative aspects such as granularity, required client knowledge and versatility.

However, in this work, we concentrate only on the comparison between the models rather than providing absolute numbers for the implementation examples we have introduced as these are not relevant in the frame of a comparative analysis on the conceptual level. Nevertheless, the examples will also be briefly discussed with respect to implementation aspects, limitations, and future work (section \ref{sec:compatibility and interoperability}).

\subsection{Modeling the Association Mechanisms as Graphs}\label{sec:graph model}

To compare the association models not only qualitatively but rather on a quantitative basis, the three association approaches need to be put into a mathematical framework. We model the association approaches as directed graphs. This seems natural because associations between FDOs and operations are all based on references pointing from one entity to another entity. Those entities might be FDOs, operations, profiles (in the case of profile typing) or attributes (in the case of attribute typing). Relations such as ``FDO $f$ contains attribute $a$ in its information record'' and ``attribute $a$ points to operation $o$'' directly translate into edges, while the entities named above translate into vertices in a graph. This is further specified by Definition \ref{def: graph models} and visualized by \autoref{fig:example graphs}.

We index our association models with $i\in \{1,2,3\}$, such that $i=1$ refers to record typing, $i=2$ to profile typing, and $i=3$ to attribute typing. In the following, we examine each association model separately under the assumption that the whole FDO ecosystem purely relies on a single association approach. 

\begin{definition}[Components]\label{def: components}
    Let $F$ be the set of all FDOs representing data, $O$ the set of all operations, $P_i$ the set of all FDO profiles and $A_i^{def}$ the set of all attribute definitions (referring to PID-Information Types), in the whole FDO ecosystem. Attribute definitions are instantiated by attributes (e.g., in FDO, operation or profile information records) which are given by the set $A_i$. We denote the numbers of those quantities by $\vert F\vert,$ $ \vert O\vert$, $\vert P_i\vert$, $\vert A_i^{def}\vert$ and $\vert A_i\vert$, respectively. The set $C_i = F\cup O \cup P_i \cup A_i^{def}$ contains all \textbf{components} of the $i$-th association model.
\end{definition}

Attribute definitions determine a key for an attribute together with a set of restrictions on the value of the attribute. Each attribute $a=(a_1,a_2)\in A_i$ is represented by a tuple that consists of a key $a_1$ and a value $a_2$. Two attributes $a=(a_1,a_2)$, $b=(b_1,b_2)\in A_i$ are considered to be the same element (i.e., $a=b$) if and only if they have the same key-value-pair (i.e., $a_1=b_1$ and $a_2=b_2$) and they are part of the same information record.

All components of the FDO ecosystem are uniquely identified by PIDs (and might be considered as FDOs themselves, depending on the maturity of implementation). Some components, such as the set of profiles and the set of attribute definitions or attributes, depend on the examined association approach. For example, attribute definitions might have different required keys and restrictions on the values depending on the model. Also, the content of the profiles might differ according to the implementation and the chosen model. Hence, the set of attributes and the set of profiles is indexed by $i\in\{1,2,3\}$. The FDOs and operations are considered to be the same sets in all models (strictly speaking, we assume that there are bijective mappings $\mathcal{M}_{ij}:F_i\rightarrow F_j$ and $\mathcal{M}'_{ij}:O_i\rightarrow O_j$ between FDOs from different models and operations from different models, for $i\neq j$).  

\begin{definition}[Graph Models]\label{def: graph models}
We define a simple graph model for the three association approaches. For $i\in \{1,2,3\}$, denote $G_i=(V_i,E_i)$ as the \textbf{graph} $G_i$ which consists of \textbf{vertices} $v_i\in V_i$ which are connected by \textbf{edges} $e_i=\{x_i,y_i\}\in E_i$ with $x_i,y_i\in V_i$.   
\begin{itemize}
    \item $i=1$: For record typing, each FDO is directly associated with an operation via an attribute within the information record. Hence, 
    \begin{align*}
        V_1=&\ F\cup A_1\cup O,\\
        E_1=&\ \{\{f,a\}: \text{FDO }f\in F \text{ has the attribute } a\in A_1 \}\  \nonumber\\ 
        &\ \cup \{ \{a,o\}: \text{attribute }a\in A_1\text{ references operation }o\in O\}.
    \end{align*}
    
    \item $i=2$: In terms of profile typing, each FDO references a profile via an attribute in the information record. In turn, an attribute in the profile information record references an FDO operation. Therefore, 
    \begin{align*}
        V_2=&\ F\cup A_2\cup P_2\cup O,\\
        \begin{split}		
        E_2=&\ \{ \{f,a\}: \text{FDO }f\in F \text{ has the attribute } a\in A_2\}\\
        &\ \cup \{ \{a,p\}: \text{attribute }a\in A_2\text{ references profile }p\in P_2\}\\
        &\ \cup \{ \{p,a\}: \text{profile }p\in P_2 \text{ has the attribute } a\in A_2\}\\
        &\ \cup \{ \{a,o\}: \text{attribute }a\in A_2\text{ references operation }o\in O\}.
        \end{split} 
    \end{align*}

    \item $i=3$: For attribute typing, each operation FDO references a set of attributes within an FDO information record (via attributes in the operation FDO). Hence, 
    \begin{align*}
        V_3=&\ F\cup A_3\cup O,\\
        \begin{split}
        E_3=&\ \{ \{o,a\}: \text{operation }o\in O \text{ has the attribute } a\in A_3\}\\
        &\ \cup \{ \{a,a'\}: \text{attribute }a\in A_3\text{ references attribute }a'\in A_3\}\\
        &\ \cup \{ \{a',f\}:\text{attribute } a'\in A_3 \text{ is contained in FDO }f\in F\}.
        \end{split}
    \end{align*}
\end{itemize}
\end{definition}

\begin{figure}
\centering
\begin{subfigure}{0.4\textwidth}
    \includegraphics[width=\textwidth]{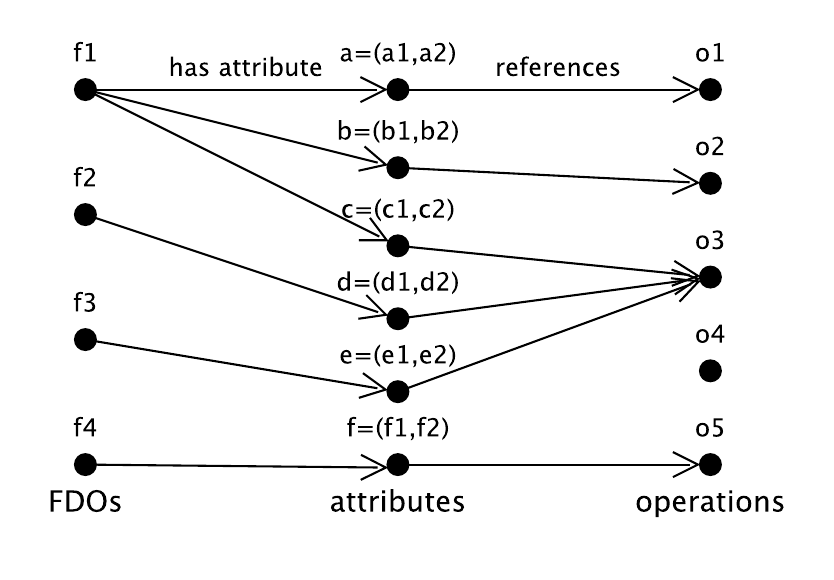}
    \caption{Record typing}
    \label{fig:graph record typing}
\end{subfigure}

\vspace{0.5cm}
\begin{subfigure}{0.65\textwidth}
    \includegraphics[width=\textwidth]{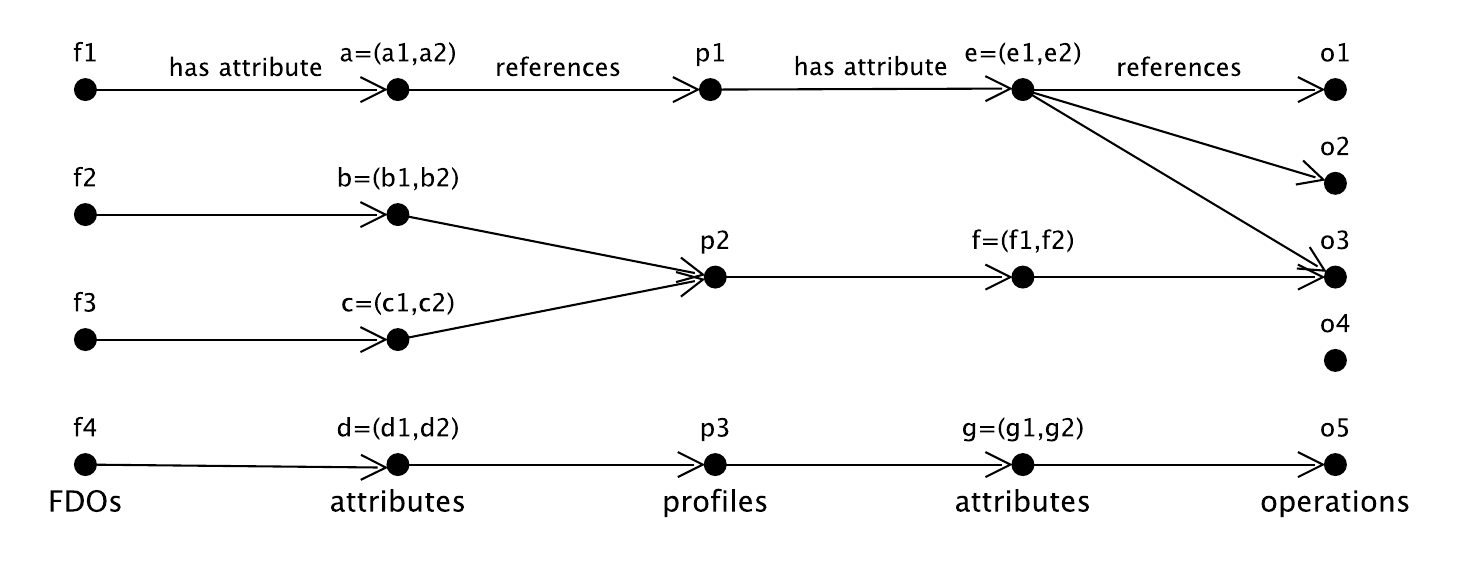}
    \caption{Profile typing}
    \label{fig:graph profile typing}
\end{subfigure}

\vspace{0.5cm}
\begin{subfigure}{0.55\textwidth}
    \includegraphics[width=\textwidth]{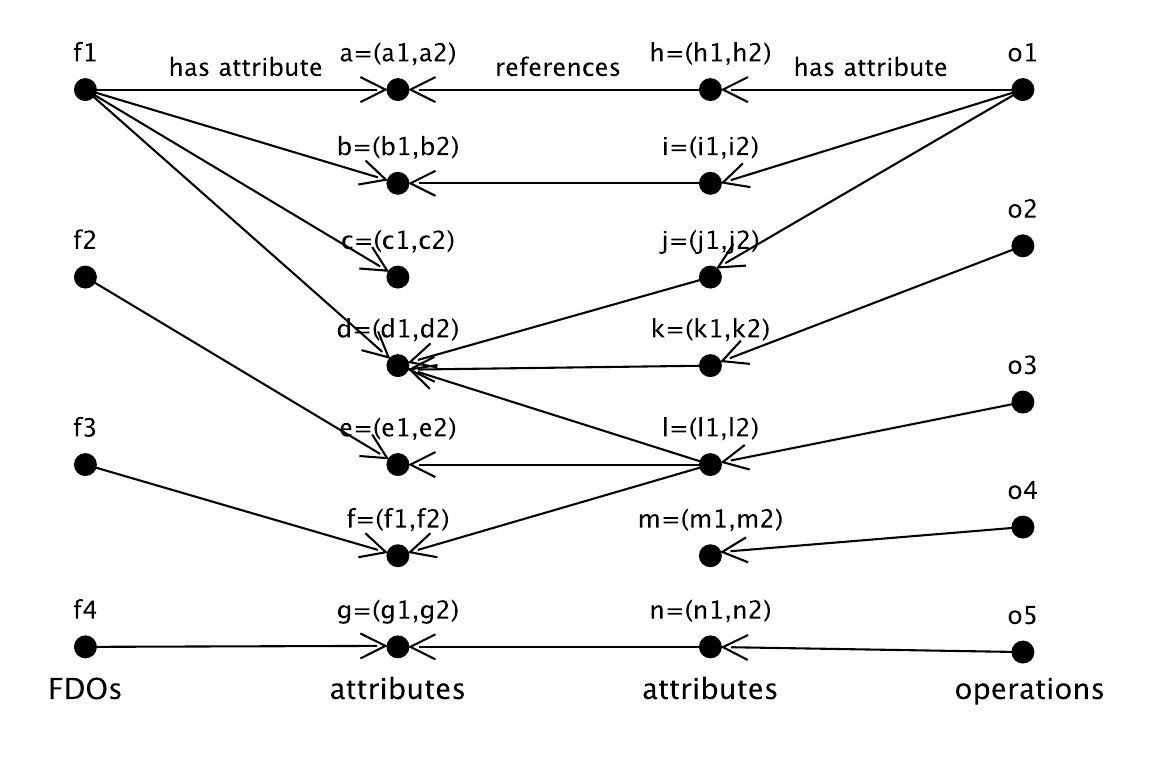}
    \caption{Attribute typing}
    \label{fig:graph attribute typing}
\end{subfigure}

\vspace{0.5cm}
\caption{Graph representations. Exemplary graphs modeling the three association approaches.}
\label{fig:example graphs}
\end{figure} 

Note that all involved attributes represent references, except the attributes for $i=3$ that are located in the FDO information record, which are instead referenced by other attributes. References point from the originating entity to the referenced entity. Consequently, edges $\{x,y\}\in E_i$ are naturally ordered and may be modeled as directed edges (see \autoref{fig:example graphs} where references are displayed as arrows explicitly). However, for any arrow pointing from $x\in V_i$ to $y\in V_i$, there will never be another arrow pointing from $y$ to $x$ due to model definition. Hence, there is no need to differentiate between the orientation of edges, so we will work with simple graphs and adopt the notation as in Definition \ref{def: graph models}.  

\autoref{fig:example graphs} visualizes Definition \ref{def: graph models}. The three graphs all contain the same FDOs $f_1,\dots,f_4$, the same operations $o_1,\dots,o_5$ and they represent the same set of associations: $f_1$ is associated with $o_1,o_2$ and $o_3$, whereas $f_2$ and $f_3$ are both associated to $o_3$, and $f_4$ is associated to $o_5$.

For record typing, each FDO might have several attributes for operation association, which contain all the same key (i.e., $a_1=b_1=c_1=d_1=e_1=f_1$). The attributes directly reference an operation via their value. In this example, attributes $c$, $d$ and $e$ all have the same value (i.e., $c_2=d_2=e_2$) because they reference the same operation. Each sequence of arrows that connects an FDO on the left side with an operation on the right side represents one FDO-operation-association.

In terms of profile typing, each FDO has exactly one attribute containing the profile reference. Those attributes have the same keys (i.e., $a_1=b_1=c_1=d_1$). If two FDOs have the same profile, their attributes point to the same profile in the graph (i.e., $b_2=c_2$). Each profile contains exactly one attribute ($e_1=f_1=g_1$) to specify a list of operations as its value. Similar to record typing, each closed sequence of arrows from left to right represents one FDO-operation-association.

For attribute typing, each data FDO may contain multiple attributes. Similarly, each operation FDO may contain multiple attributes, with keys being all tantamount (i.e., $h_1=i_1=j_1=k_1=l_1=m_1=n_1$). The attributes in the operation FDO information record reference attributes in the data FDO information record (i.e., in this example, we have $h_2=a_1$, $i_2=b_1$, $j_2=d_1$, and so on). If all arrows originating from an operation $o_i$ form closed lines with arrows originating from $f_j$ (where the end of the arrow corresponds to the attributes of the data record), then $o_i$ is associated with $f_j$. Note that this model is a simplification of attribute typing because we just consider the case that attributes in the operation record match with attributes in the FDO information record if the attribute in the FDO information record is present (i.e., has the desired key). We do not consider possible restrictions on the allowed values of attributes in the FDO information record and the resulting impact on granularity.

\subsection{Evaluation of Quality Indicators and Metrics}\label{sec:quality measures}

We are going to examine quantitative quality indicators (simplicity, efficiency, flexibility) and qualitative aspects (granularity, required client knowledge and versatility). For the quantitative quality indicators, we define simple mathematical measures which are separately evaluated for each model under the assumption that the whole FDO ecosystem purely relies on a single association approach.

Throughout this work, we use big \(\mathcal{O}\) notation to assess computational complexity of the conceptual models. Note that we generally make no assumptions about the data structure used in an implementation in which the information concerning the assessments would be stored.
\bigskip

\textbf{Quantitative Quality Indicators}

\emph{Simplicity} refers to how complex it will be for a client to handle an FDO ecosystem that applies a given association model with respect to its structure. This can be measured using metrics such as the number of components involved and the number of their relations using graph theoretical methods. 

\emph{Efficiency} takes into account how complex it will be to find all operations that are associated to an FDO or to assess whether a certain FDO is associated to a given operation. This can be measured using metrics such as the number of edges in the graph that make up an association.

\emph{Flexibility} as a quality indicator relates to the question how many active modifications are required when new components are added to an existing FDO ecosystem that applies a particular association model. This can be measured using metrics such as the number of updates that must be performed when a new association between an FDO and an operation is made.
\bigskip

\textbf{Qualitative Aspects}

Qualitative aspects we are going to consider are the \emph{granularity} of the association models in comparison to the \emph{amount of client knowledge} that is required in order to add the desired associations to a new FDO. In addition, the \emph{versatility} of the models is discussed, which considers the possible processing options of an FDO through its associated operations in relation to the aspects imposed by \emph{efficiency} and \emph{flexibility}.
\bigskip

\begin{definition}
The following notation is introduced to evaluate the quantitative measures (see Theorem \ref{theorem: evaluated measures}).
\begin{enumerate}
    \item For any non-empty subset $F'\subseteq F$, $O_{F'}$ is the set of all operations that are associated with at least one FDO $f\in F'$. For the set containing a single element $F'=\{f\}$, we write $O_f$ instead of $O_{\{f\}}$. For profile typing ($i=2$) and a non-empty subset of profiles $P'\subseteq P_2$, we define the set of all operations that are referenced by at least one profile $p\in P'$ as $O_{P'}$.
    
    \item For any set $O'\subseteq O$, $F_{O'}$ is the set of all FDOs that are associated with at least one operation $o\in O'$. For the set containing a single element $O'=\{o\}$, we write $F_o$ instead of $F_{\{o\}}$.
    
    \item For $f\in F$ and  $o\in O$, let $A_f$ and $A_o$ be the sets of all attributes in the information records of FDOs and operations respectively. 
    
    \item Finally, the following definition just holds for $i=2$: For subsets $F'\subseteq F$ and $O'\subseteq O$, we define  $P_{2,F'}$ as the set of all profiles that are referenced by at least one FDO $f\in F'$, $P_{2,O'}$ as the set of profiles being associated to at least one operation $o\in O'$, and $P_{2,F'O'}=P_{2,F'}\cap P_{2,O'}$ as the set of all profiles that are being part of at least one FDO-operation-association between set elements of $F'$ and $O'$.
\end{enumerate}
\end{definition}

Note that the total number of FDO-operation-associations is represented by $\sum_{f\in F}\vert O_f\vert =\sum_{o\in O}\vert F_o\vert$ irrespective of the association model. 

\begin{definition}[Measures for Quantitative Quality Indicators]\label{def: quality indicators}
    For $i\in\{1,2,3\}$, we define the following metrics to assess the quality indicators:
    \begin{enumerate}
        \item $\mathcal{C}_i$ is the total number of components (see Definition \ref{def: components}) in the FDO ecosystem that are (potentially) part of each association mechanism. This does not only include those FDOs, operations, attribute definitions and profiles which are actually part of at least one FDO-operation-association, but the total sets of the components that might be involved. 
        \item $\mathcal{A}_i$ is the total number of instantiated attributes that are present in FDO, profile, or operation information records, which are actually part of the association mechanism. Here, attributes are counted multiple times if the same key-value pair is present in multiple information records.
        \item $\mathcal{Q}_i$ is is an upper bound on the time complexity to decide whether an FDO $f\in F$ is associated to an operation $o\in O$.
        \item $\mathcal{R}_i$ is an upper bound on the time complexity to find all FDOs that are associated with a single operation.
        \item $\mathcal{S}_i$ is an upper bound on the time complexity to find all operations associated with a single FDO.
        \item $\mathcal{T}_i$ is an upper bound on the time complexity to perform all required updates in the FDO ecosystem to associate a new operation with a set $F'\subseteq F$ of FDOs.
        \item $\mathcal{U}_i$ is an upper bound on the time complexity to perform all required updates in the FDO ecosystem to associate a new FDO with a set of operations $O'\subseteq O$.
    \end{enumerate}
\end{definition}

\begin{theorem}[Evaluated Measures]\label{theorem: evaluated measures}    
    The measures specified in Definition \ref{def: quality indicators} are evaluated to the following quantities:
    \begin{enumerate}
        \item 
        \begin{align*}
            \mathcal{C}_1&=\vert F\vert + \vert O\vert + 1, \\
            \mathcal{C}_2&= \vert F\vert + \vert O\vert + \vert P_2\vert + 2, \\
            \mathcal{C}_3&= \vert F\vert + \vert O\vert + \vert A_3^{def}\vert.
        \end{align*}

        \item For $i=3$, let $b_1,\dots,b_{\vert F_O\vert}\in \mathbb{N}$ be the number of attributes being part of the association mechanism for the FDOs $f_1,\dots, f_{\vert F_O\vert}$, and let $d_1,\dots,d_{\vert O_F\vert}\in \mathbb{N}$ be the number of attributes taking part in the association mechanism for each operation $o_1,\dots, o_{\vert O_F\vert}$. 
        \begin{align*}
            \mathcal{A}_1&=\displaystyle\sum_{f\in F_O}\vert O_f\vert,\\
            \mathcal{A}_2&=\vert F_O\vert + \vert P_{2,FO}\vert ,\\
            \mathcal{A}_3&=\displaystyle\sum_{j=1}^{\vert F_O\vert}{b_j} + \displaystyle\sum_{j=1}^{\vert O_F\vert}{d_j}.  
        \end{align*}

        \item \begin{align*}
            \mathcal{Q}_1 &= \mathcal{O}(\vert A_f\vert),\\
            \mathcal{Q}_2 &= \mathcal{O}(\vert A_f\vert + \vert O_{P_f}\vert),\\
            \mathcal{Q}_3 &= \mathcal{O}(\vert A_f\vert +\vert A_o\vert ).
        \end{align*}
        
        \item 
        \begin{align*}
            \mathcal{R}_1&=\mathcal{O}\left(\displaystyle\sum_{f\in F}\vert A_f\vert\right),\\
            \mathcal{R}_2&=\mathcal{O}\left(\displaystyle\sum_{f\in F}\vert A_f\vert+ \displaystyle\sum_{p\in P_{F}}\vert O_{p}\vert\right),\\
            \mathcal{R}_3&=\mathcal{O}\left(\displaystyle\sum_{f\in F}\vert A_f\vert+\vert A_o\vert\right).
        \end{align*}

        \item 
        \begin{align*}
            \mathcal{S}_1&= \mathcal{O}(\vert A_f\vert),\\
            \mathcal{S}_2&=\mathcal{O}(\vert A_f\vert+\vert O_{P_f}\vert),\\            \mathcal{S}_3&=\mathcal{O}\left(\vert A_f\vert + \displaystyle\sum_{o\in O}\vert A_o\vert\right).\\
        \end{align*}

        \item 
        \begin{align*}
            \mathcal{T}_1&=\mathcal{O}(\vert F'\vert),\\
            \mathcal{T}_2&= \mathcal{O}(\vert P_{2,o}\vert),\\
            \mathcal{T}_3&=0.
        \end{align*}

        \item 
        \begin{align*}
            \mathcal{U}_1&=\mathcal{O}(\vert O'\vert),\\
            \mathcal{U}_2&=0,\\
            \mathcal{U}_3&=0. 
        \end{align*}
    \end{enumerate}
\end{theorem}

\begin{proof}
    \begin{enumerate}
        \item According to Definition \ref{def: components}, the components involve the sets $F$, $O$, $P_i$ and $A_i^{def}$. However, we just count those components which are potentially taking part in the association mechanism. For $i=1$, this is the set of FDOs, the set of operations, and a single attribute definition (as all FDOs reference their operations via the same attribute key). 
        For $i=2$, there are two attribute definitions involved in the association mechanism, one to reference an FDO profile in all FDO information records, and one to reference a list of operations in all profile information records. For $i=3$, there are no restrictions on the set of attributes which are being used in the FDO information records. Hence, all attribute definitions $A_3^{def}$ are potentially taking part in the association mechanism.

         \item Counting the number of attributes being part of the association mechanism means to count all arrows with the label ``has attribute'' as illustrated in \autoref{fig:example graphs} that are part of at least one FDO-operation association. For $i=1$, each association corresponds to one attribute, such that the number of attributes equals the total number of associations. For $i=2$, each FDO contains exactly one attribute to be connected to a profile (totaling $\vert F_O\vert$ attributes), and each profile has exactly one attribute that connects it to a set of operations (totaling $\vert P_{2,FO}\vert$ attributes). For $i=3$, the equation follows by definition of $b_j$ and $d_j$.

        \item Let $f\in F$ be any FDO and $o\in O$ be any operation. For $i=1$, a client would need to search the whole FDO information record for the attribute containing the reference to $o$, taking time $\mathcal{O}(\vert A_f\vert)$. For $i=2$, one needs to find the profile $p$ in the FDO information record within time $\mathcal{O}(\vert A_f\vert)$. Since accessing the profile and its list of operations highly depends on the implementation but is not directly relevant for the comparative analysis, we assume access in constant time. Finally, the list of operations needs to be searched for the reference to $o$, taking time $\mathcal{O}(\vert O_{p}\vert)$. For $i=3$, additionally, all attributes in the operation information record need to be found that determine the association, which is done in $\mathcal{O}(\vert A_o\vert)$. Afterwards, each of the associations that were found need to be matched against the attributes in the information record (after converting either the attributes in the FDO or the attributes in the operation FDO into a suitable format).

        \item For $i=1$, one has to search each FDO information record in the FDO ecosystem for its operations, which is $\mathcal{O}(\sum_{f\in F}\vert A_f\vert)$. For $i=2$, similar time is required to find all profiles $P_{2,F}$. A profile has one attribute containing a list of operations, and we assume that each list of operations can be accessed in constant time. Furthermore, checking whether those lists contain the operation requires reading the whole operation list within time $\mathcal{O}(\sum_{p\in P_F}\vert O_{p}\vert)$. For $i=3$, first find all operation attributes and convert them into a suitable format within time $\mathcal{O}(\vert A_o\vert)$. Then, read all attributes in all FDOs and check whether they match the operation attributes, taking time $\mathcal{O}(\sum_{f\in F}\vert A_f\vert)$.

        \item For all $i\in \{1,2,3\}$, it is required to read all attributes in the information record. For $i=2$, one then accesses the profile $p$ within $\mathcal{O}(1)$ and the list of operations also within $\mathcal{O}(1)$. Reading all elements from that list takes time $\mathcal{O}(\vert O_p\vert )$. For $i=3$, the FDO information record is converted into a suitable format (within $\mathcal{O}(\vert A_f\vert)$). Then, each operation FDO has to be checked against the data FDO, which requires time $\mathcal{O}(\sum_{o\in O}\vert A_o\vert)$.

        \item For $i=1$, relating a new operation to the set $F'$ requires to add one attribute in each FDO information record, yielding $\mathcal{O}(\vert F'\vert)$ updates in total. For $i=2$, the new operation needs to be added to all profiles that it should be applicable to, which are $\vert P_{2,o}\vert$. For $i=3$, no updates need to be done because the set $F'$ is implicitly defined by the attributes in the operation FDO.

        \item  To associate a new FDO with a set of operations $O'$, $\vert O'\vert$ new attributes need to be added to the FDO information record for $i=1$. In the case of $i=2$, no updates need to be performed because the new FDO is required to have a profile anyway and the profile implicitly defines the set $O'$. For $i=3$, no updates need to be performed with the same reason as detailed in 6.  
    \end{enumerate}    
\end{proof}

Note that the set $F'$ in part 6 is either defined by the client ($i=1$) or it is imposed by the model ($i=2$ and $i=3$). This is because the three association mechanisms follow different ideas: For $i=1$, the client can decide on any association individually, so he will define the set $F'$. For $i=2$, when a new operation is added, the associations are partly to be decided on by somebody who has the right to edit the required profiles and partly implied by the model itself (the associations between profiles and FDOs are already given and cannot be changed). For $i=3$, the set $F'$ is fully determined by the model in advance, depending on the attributes that are specified in the operation record. A similar observation applies to part 7: For $i=1$, the client will define the set $O'$, whereas for $i=2$ and $i=3$, the set $O'$ is fully specified by the model. Such considerations need to be taken into account when assessing the quality measures.

\subsection{Comparison of Measures}

We now compare the measures to evaluate the strengths and weaknesses of the different association models, starting with the quantitative measures.
\begin{itemize}
    \item \textbf{Simplicity:} Both the number of components and the number of attributes that are being part of the association mechanism are measures for the simplicity of the model. If few attributes are being involved, the information records (of FDOs, profiles, and operations) can be kept comparatively short. If additionally few components are involved, the models are easier to understand for potential users. Regarding components, we have $\mathcal{C}_1<\mathcal{C}_2\leq \mathcal{C}_3$. However, there does not seem to be any general order of $\mathcal{A}_1,\mathcal{A}_2$ and $\mathcal{A}_3$. Under the assumption that the number of profiles is relatively small in comparison to the number of FDOs, in most cases, $\mathcal{A}_2$ will be smallest. $\mathcal{A}_1$ would be even smaller in the special case that each FDO is associated to exactly one operation. $\mathcal{A}_3$ is comparably small if each operation relies on few attributes for association. Overall this shows that record typing and profile typing are less complex models in comparison to attribute typing.
    
    \item \textbf{Efficiency:} All measures $\mathcal{Q}_i$, $\mathcal{R}_i$ and $\mathcal{S}_i$ quantify the effort for a client to find certain FDO-operation-associations within the FDO ecosystem. To compare those measures, we note that all upper bounds are sharp upper bounds.
    
    $\mathcal{Q}_i$ quantifies the effort to decide whether a certain FDO is associated to an operation. This is obviously smallest for record typing, while $\mathcal{Q}_2\lesssim \mathcal{Q}_3$ and $\mathcal{Q}_2\gtrsim \mathcal{Q}_3$ are both possible: if there are few operations associated with the profile of $f$ for $i=2$, then $\mathcal{Q}_2\lesssim \mathcal{Q}_3$, otherwise $\mathcal{Q}_2\gtrsim \mathcal{Q}_3$.
    Considering $\mathcal{R}_i$, it is obvious that $\mathcal{R}_1$ is smallest. For the other two association models, it is both possible that $\mathcal{R}_2\lesssim \mathcal{R}_3$ or $\mathcal{R}_2\gtrsim \mathcal{R}_3$: If (very) few operations are associated with the set $F$ (for example, when all $f\in F$ have the same profile), then $\mathcal{R}_2\lesssim \mathcal{R}_3$. Otherwise, $\mathcal{R}_2\gtrsim \mathcal{R}_3$.
    For $\mathcal{S}_i$, we observe that $\mathcal{S}_1$ is smallest, while $\mathcal{S}_2$ also scales with the number of operations related to the given profile and $\mathcal{S}_3$ scales with the number of attributes in all operations, which is considerably larger.

    Overall, this shows that record typing is the best approach in terms of efficiency. Profile typing and attribute typing are less efficient in terms of measures $\mathcal{Q}_i$ and $\mathcal{R}_i$. However, the measure $\mathcal{S}_3$ reveals the high costs of attribute typing in comparison to the other models, because one has to iterate over all attributes of all operations in the FDO ecosystem to find all operations associated to one FDO.
    
    \item \textbf{Flexibility:} Assuming that the number of FDOs associated to the new operation is much larger than the number of profiles (for $i=2$) associated to this operation, it trivially follows that $\mathcal{T}_3\lesssim \mathcal{T}_2 \lesssim \mathcal{T}_1$. For $\mathcal{U}_i$, obviously $\mathcal{U}_1\gtrsim \mathcal{U}_2=\mathcal{U}_3$. Hence, in terms of required updates, attribute typing is most efficient, followed by profile typing. In comparison, record typing is relatively inefficient.
\end{itemize}

Finally, we will comment on qualitative aspects, i.e., the granularity and client knowledge, as well as the versatility. 
\begin{itemize}
    \item \textbf{Granularity and client knowledge:} For record typing, each FDO can be associated with any operation as desired by the client. This is the most granular approach, as any combination of FDOs and operations is possible. However, for each newly defined FDO, the client who has introduced the FDO information record has to think of which operations to include into the information record. This requires both domain knowledge regarding the content information in the FDO, and knowledge about the association mechanism. Therefore, the price for higher granularity is that for each new FDO a careful individual inspection might be required to make an informed decision on operation association. 
    
    Attribute typing has a slightly smaller granularity, as not every FDO can be seamlessly associated to any operation. In turn, the association mechanism works out automatically, which means that clients just need to include all information they have available into the information record, without deciding for specific operations or attributes. However, in case a client has a specific operation in mind that was not automatically associated to the FDO but which he wants to be associated, he still needs to figure out which additional attributes to include into the FDO information record. 
    
    Profile typing is the least granular approach. As each operation is associated to a whole class of FDOs, there is a need for many different profiles to be made available to the client in order to reach a granularity that is comparable with the other models. The advantage of profile typing is that the client just has to make an informed choice which profile to use, and then he will be instructed which attributes are required in the information record due to the profile definition. Hence, he does not have to think at all about associating his FDO to any operations. 

    \item \textbf{Versatility:} In contrast to its high granularity, the overall versatility of record typing is considered to be the lowest, as each FDO-operation association must be explicitly declared to increase the possible processing options for an FDO.
    
    Profile typing has much greater versatility compared to record typing because a profile is typically reused several times to create a set of FDOs, and all of these FDOs automatically have the possible processing options defined by the operations associated with that profile. 
    
    Compared to attribute typing, the versatility of profile typing is potentially less because attribute definitions that constitute an association condition are typically reused across profiles and may occur in multiple data FDOs. These FDOs then automatically have the possible processing options defined by these operations.
    In this way, an operation can still be associated with any FDO whose profile contains the required set of attributes, and the association is not missed simply because the association between the operation and the profile was not explicitly made. In addition, an operation associated via profile typing may assume the presence of specific, not necessarily mandatory, attribute definitions in the profile. This could result in incompatibilities when executing the operation in case these attribute definitions were not instantiated for a particular FDO. With attribute typing, this cannot happen since the instantiation of all required attribute definitions is assured as part of the association process.
\end{itemize}

\subsection{Interoperability of Association Models}\label{sec:compatibility and interoperability}

In contrast to the other quality indicators, we do not define specific metrics to quantify different levels of interoperability which is out of scope for this work. Instead, we describe the implications for interoperability of FDOs and their operations by the compatibility of the introduced association models.

Interoperability of FDO operations refers to the ability to perform consistent, standardized operations on FDOs across different systems and platforms, ensuring that actions such as accessing, processing, or transforming the objects can be executed reliably and requested uniformly by a client, regardless of the environment. Therefore, from our point of view, different association models should be consistent and compatible with respect to a standardized FDO type system that utilizes one or more typing mechanisms.

This ensures that when an FDO is accessed or manipulated across different platforms, its type definitions and associated operations are consistently interpreted and executed. The type system provides a common language for using the standardized structure of FDOs based on one or more typing mechanisms, enabling seamless interaction between systems. Regardless of which typing mechanisms are implemented within a service, it is critical that all clients accessing an FDO obtain the same set of associated operations independent from the underlying model.

Profiles are essential for this, since they provide a minimal, standardized metadata structure for all FDOs. Because all association models for FDO Operations are expected to work with either a profile, profile attributes (also operations that are specified in the record), or a combination of both, they maintain a basic level of compatibility. Different association models using either record, profile or attribute typing can therefore be applied to the same FDO since the type system serves as a common language that allows different models to ``understand'' which operations are valid. This also means that regardless of whether a simple or complex association model is used, the kernel metadata ensures that at least a core set of operations can be applied universally.

\subsection{Implementation Considerations}
Based on the results of our comparative evaluation and considering the approach of modeling the different association models using directed graphs, we conclude that implementations such as described in section \ref{sec:Record Typing in Interactive Computing Environments}-\ref{sec:Attribute Typing with Operation FDOs} would highly profit from storing the interconnected components and the rules of each association model in proper graph data structures. This way, the assessment of the information about associations as described by the metrics is done once at ingestion time and stored as vertices and nodes according to the model, yielding the structure as illustrated in \autoref{fig:example graphs}. The repetitive procedures described by the quality indicators and quantified by their metrics are then facilitated. For example, assessment about which operations are associated with a given FDO and vice versa can be performed with simple graph queries. More complex procedures such as it is the case for attribute typing (cf. \(\mathcal{S}_{3}\)) could be also compensated this way by caching and integrating rules for inferring information. This could take place on the level of the object entities, as well as on the level of the services which store additional information about some components, for example the profiles.

\section{Conclusions}
In this paper, we explored multiple modeling approaches for associating FAIR Digital Objects with their operations through different typing mechanisms based on three example implementations. Our analysis demonstrates that each model—record typing, profile typing, and attribute typing—has distinct advantages and trade-offs concerning simplicity, efficiency, flexibility, versatility, and interoperability for FDO ecosystems. While record typing offers simplicity, profile typing and attribute typing provide enhanced flexibility and versatility. Our findings also indicate that these association models are in so far compatible with each other that a particular FDO entity could incorporate all approaches at the same time. This is also relevant with respect to interoperability between different FDO ecosystems. Ultimately, adopting an association model will depend on the specific requirements of the data environment, including client expectations and computational constraints. Future work will need to consider how to manage FDO ecosystems at scale and which technologies are most suitable for implementing different models, ensuring a robust foundation for machine-actionable data infrastructures.

\paragraph{Acknowledgements}
This work is funded by the Helmholtz Metadata Collaboration Platform (HMC), NFDI4ING (DFG – project number 442146713), and supported by the research program “Engineering Digital Futures” of the Helmholtz Association of German Research Centers. Funded by the European Union. This work has received funding from the European High Performance Computing Joint Undertaking (JU) under grant agreement No 101093054.

We would like to thank the FDO Forum group participants who contributed valuable insights through discussions and remarks during this work. Special thanks go to Larry Lannom from the Corporation for National Research Initiatives (CNRI) and Yudong Zhang from GESIS--Leibniz-Institut für Sozialwissenschaften.

\paragraph{Authors' Contributions}
Supervision: Nicolas Blumenröhr; Conceptualization: Nicolas Blumenröhr, Jana Böhm, Marco Kulüke, Christophe Blanchi, Peter Wittenburg, Ulrich Schwardmann, Sven Bingert; Methodology and Evaluation: Nicolas Blumenröhr, Jana Böhm, Philipp Ost; Revision of State-of-the Art analysis and background: Nicolas Blumenröhr, Jana Böhm, Philipp Ost, Peter Wittenburg, Ulrich Schwardmann, Sven Bingert, Christophe Blanchi
\printbibliography

\newpage

\theendnotes
\end{document}